\documentclass[12pt,reqno,a4paper]{amsart}
\usepackage{amsmath,amssymb,amsfonts,amscd}
\usepackage[mathscr]{eucal}
\usepackage{hyperref,xcolor} 
\usepackage{comment,textcomp}

\date{22 August (4 September)  2019}

\def\co{\colon\thinspace}

\newtheorem{theorem}{Theorem}
\newtheorem{lemma}{Lemma}
\newtheorem{proposition}{Proposition}
\newtheorem{corollary}{Corollary}

\theoremstyle{definition}
\newtheorem{definition}{Definition}
\newtheorem{example}{Example}
\newtheorem{remark}{Remark}

\newcommand{\C}{\mathbb{C}}

\newcommand{\ZZ}{{\mathbb Z}_2}

\newcommand{\f}{{\varphi}}

\newcommand{\p}{\partial}

\newcommand{\g}{\gamma}

\newcommand{\der}[2]{{\frac{\partial {#1}}{\partial {#2}}}}

\DeclareMathOperator{\Ber}{Ber}

\DeclareMathOperator{\str}{str}

\DeclareMathOperator{\SpO}{SpO}

\DeclareMathOperator{\Spin}{Spin}

\DeclareMathOperator{\SO}{SO}
\DeclareMathOperator{\id}{id}

\newcommand{\lder}[2]{{\partial {#1}/\partial {#2}}}

\newcommand{\oder}[2]{{\frac{d {#1}}{d {#2}}}}

\newcommand{\RR}{\mathbb R}
\newcommand{\fun}{C^{\infty}}

\newcommand{\e}{{\varepsilon}}

\newcommand{\F}{{\Phi}}

\newcommand{\const}{\mathrm{const}}

\newcommand{\kt}{{\tilde k}}

\newcommand{\itt}{{\tilde\imath}}

\newcommand{\at}{{\tilde a}}
\newcommand{\bt}{{\tilde b}}

\newcommand{\qt}{{\tilde q}}

\renewcommand{\mod}{\,\mathrm{mod}\,}

\newcommand{\tto}{{\linethickness{2pt}
		  \,\begin{picture}(1,0)
                   \put(0,0.26){\line(1,0){0.95}}
                   \put(0,0){$\boldsymbol{\rightarrow}$}
                  \end{picture}
                  }\,
}

\newcommand{\oto}{{\linethickness{0.5pt}
		  \,\begin{picture}(1,0)
		  \put(0.07,0.175){\line(0,1){0.2}}
                   \put(-0.01,0){$\boldsymbol{\Rightarrow}$}
                  \end{picture}
                  }\,
}

\newcommand{\ttto}[1]{\stackrel{#1}{\vphantom{\rightrightarrows}\tto}}

\newcommand{\ttoh}{\tto_{\hbar}}

\newcommand{\Dbar}{{\mathchar '26 \mkern-10 mu D}}

\newcommand{\Hqu}{H^{\hbar}\Bigl(\overset{\,1}{x\vphantom{\frac{\hbar}{i}}},\overset{2}{\frac{\hbar}{i}\der{}{x}}\Bigr)}

\newcommand{\Sinf}{S_{\infty}}

\newcommand{\Linf}{L_{\infty}}

\DeclareMathOperator{\funn}{\mathbf{C^{\infty}}}
\DeclareMathOperator{\pfunn}{\mathbf{\Pi\!C^{\infty}}}

\title[Thick morphisms   and spinor representation]{Thick morphisms of supermanifolds, quantum mechanics, and spinor representation}
\author{Hovhannes~Khudaverdian}
\address{Department of Mathematics,  University of Manchester, Manchester,   UK}
\email{khudian@manchester.ac.uk}
\author{Theodore~Voronov}
\address{Department of Mathematics,  University of Manchester, Manchester,   UK}
\email{theodore.voronov@manchester.ac.uk}
\address{%
{Faculty of Physics, Tomsk State University, Tomsk, 634050, Russia}}

\begin{document}
\begin{abstract}
``Thick'' or ``microformal'' morphisms of supermanifolds generalize ordinary maps. They were discovered as a tool for homotopy algebras. Namely, the corresponding pullbacks provide   $L_{\infty}$-morphisms  for $S_{\infty}$  or Batalin--Vilkovisky algebras.  It was clear from the start that constructions used for thick morphisms closely resemble   some fundamental notions in quantum mechanics and their classical limits (such as action, Schr\"{o}dinger and Hamilton--Jacobi equations, etc.) There was also a natural question about any connection of thick morphisms with spinor representation. We answer both questions here.
We establish   relations  of thick morphisms with fundamental concepts of quantum mechanics. We also show that in the linear setup  quantum thick morphisms  with quadratic action give (a version of)   the   spinor representation   for a   certain category of canonical linear relations, which is an analog of the  Berezin--Neretin     representation and a generalization of the metaplectic representation (and ordinary spinor representation).
\end{abstract}

\maketitle
\tableofcontents

\section{Introduction and preliminaries}\label{sec.intro}

\subsection{Introduction}
Thick (microformal) morphisms of supermanifolds were discovered as a tool for constructing $L_{\infty}$-morphisms of homotopy algebras such as $S_{\infty}$- or $P_{\infty}$-algebras. Their quantum version  does the same for Batalin--Vilkovisky  algebras. See~\cite{tv:nonlinearpullback,tv:oscil,tv:microformal}, also~\cite{tv:gradedmicro} and ~\cite{tv:qumicro,tv:tangmicro}. (The starting point was a problem concerning higher Koszul brackets introduced in \cite{tv:higherpoisson}. See also~\cite{tv:highkosz}.) We recall  definitions and main facts in the next subsection.

Construction used for   
thick morphisms have remarkable resemblance with fundamental notions of classical and quantum mechanics such as classical and quantum action.  In this paper, we try to   show that it is more than just a resemblance.

The non-linear pullback of functions by a classical thick morphism, which is key for application to homotopy structures,  is, from the viewpoint of the ambient cotangent bundles,  action on functions on   Lagrangian submanifolds. Loosely, it is an ``action on functions of $n$ variables by a transformation of a $2n$-dimensional space''. This is literary true for the case of one manifold $M^n$: then  a   ``thick diffeomorphism'' $\F\co M\tto M$ can be put into a bijection with a (formal) canonical transformation of $T^*M$  and hence there is an action  on functions on $M$ of canonical transformations of $T^*M$. This strongly resembles spinor representation, if   one recalls that the spinor representation (in the orthogonal or symplectic settings) can be seen as action of linear transformations of a ``large'' space on objects such as functions or half-forms that live on a (half-dimensional) maximally isotropic subspace. (Details depend on a particular setting and some choices e.g. of a real or holomorphic realization.)

So it is natural to ask, as we did in~\cite{tv:microformal}, whether there is an actual link between thick morphisms and spinor representation. We are  able to give a positive answer here.

Namely, here we show that pullback  by \emph{quantum} thick morphisms can be seen as a generalization of spinor  representation\,---\,in the sense that for the special case of vector spaces, a particular class of quantum thick morphisms gives a projective representation  of a   category of canonical linear relations. It is close to the spinor representation introduced by Neretin~\cite{neretin:spinor1989,neretin:categories} (as a generalization of Berezin's construction~\cite{berezin:second}). The representation given by quantum thick morphisms differs from the Berezin--Neretin representation by a multiplier. For a single vector space, it gives, also up to a multiplier, the metaplectic (or Shale--Weil or symplectic spinor) representation. These multipliers correspond, roughly, to choices of ordering in the quantization construction. Since we do everything in the super setting, we  in fact obtain representations of super categories and supergroups such as the symplectic-orthogonal supergroup $\SpO$. The case of ordinary (orthogonal) spinors   corresponds to purely odd vector spaces. (Strictly speaking, we obtain pseudo-euclidian spinors for a particular signature, namely $(m,m)$, which admits real Lagrangian subspaces.) As for   non-linear pullbacks by classical thick morphisms, which came into being   for the needs of homotopy algebras and their $\Linf$-morphisms, they can be seen  morally  as  a  generalization of ``the classical limit  of  the spinor representation'' (for   $\hbar\to 0$). It is also worth pointing out that it was Fock's work~\cite{fock:canon1959-69} first treating the connection between canonical transformations of classical  and quantum mechanics (indicated heuristically by Dirac~\cite{dirac:pqm1947}).

The  structure of the paper is as follows. In the next subsection~\ref{subsec.main}, we recall  definitions and main facts concerning classical and quantum thick morphisms. In Section~\ref{sec.hamjac}, we consider one-parameter families of thick morphisms and show how Hamilton--Jacobi and Schr\"{o}dinger equations appear in this setting. In Section~\ref{sec.spin}, we establish   connection with the spinor representation.

Throughout the paper, we use   standard language of supergeometry. In many cases, we do not particularly emphasize that we consider super objects, very often referring e.g. to ``supermanifolds'' as just ``manifolds'', etc.

\subsection{Main notions}\label{subsec.main}
Let us recall the notions of classical and quantum thick morphisms. A note on terminology: the word ``thick'' is used because our constructions give a certain ``thickening'' (in several senses) of the ordinary smooth maps. Another adjective which is applied is ``microformal''~\cite{tv:microformal}, because of    using  cotangent bundles similarly to that in microlocal analysis and of the role played by formal power expansions.

\subsubsection{Classical thick morphisms}
Let $M_1$ and $M_2$ be two manifolds or supermanifolds. (Being ``super'' is not important for the main constructions, but becomes important when we turn to applications such as to homotopy algebras and, as will appear in this paper, to spinors.) We shall also work with their cotangent bundles $T^*M_1$ and $T^*M_2$.

A \emph{thick} (also known as \emph{microformal}) \emph{morphism} $\F\co M_1\tto M_2$ is a (particular type, formal) canonical relation in $T^*M_1 \times T^*M_2$ considered with the symplectic form which is the difference of the canonical symplectic forms on $T^*M_1$ and $T^*M_2$.

In contrast with the usual perception of relations as  generalized maps  (which in our case would be between $T^*M_1$ and $T^*M_2$),  in what follows we want to see a relation $\F$ as   a kind of ``mapping'' or ``morphism'' between the manifolds $M_1$ and $M_2$ themselves, not  their cotangent bundles. The key construction to be introduced shortly will be that of pullback of functions on $M_2$ by a thick morphism $\F\co M_1\tto M_2$ taking them to functions on $M_1$ and nonlinear (unlike the familiar ordinary pullback).

We require that a  thick morphism $\F\co M_1\tto M_2$ is specified in local coordinates on $M_1$ and $M_2$ by a \emph{generating function} $S$ of a special form: it is a function of position variables on the source manifold and momentum variables on the target manifold; with respect to the latter, it is a formal power series:
\begin{equation}\label{eq.S}
    S(x,q)=S^0(x)+\f^i(x)q_i +\frac{1}{2}\,S^{ij}(x)q_jq_i +\ldots
\end{equation}
(Because of that, a thick morphism is characterized as a ``formal'' canonical relation.)
The function $S(x,q)$ defines the relation by the formula
\begin{equation}\label{eq.reldif}
    q_idy^i-p_adx^a = d\bigl(y^iq_i-S(x,q)\bigr)
\end{equation}
or
\begin{equation}\label{eq.relder}
    p_a=\der{S}{x^a}(x,q)\,, \quad y^{i}=(-1)^{\itt}\der{S}{q_i}(x,q)\,.
\end{equation}
Throughout the paper we will be using the following notations: $x^a$ and $y^i$  for coordinates on $M_1$ and $M_2$ and $p_a$ and $q_i$ for the corresponding canonically conjugate momenta (i.e. fiber coordinates in $T^*M_1$ and $T^*M_2$). Warning: $q_i$ is used for a momentum variable (on $M_2$), not a position variable! When working with supermanifolds, we use the tilde for denoting   parity ($\ZZ$-grading) of an object; in particular, tensor indices carry the parities of the corresponding coordinates (e.g. $\itt=\tilde y^i=\text{parity}(y^i)$).

The function $S$ is regarded as part of structure (so strictly speaking, a thick morphism is more than just a relation; it is a ``framed'' relation endowed with a choice of integration constant contained in $S$). It is clear that generating functions such as $S(x,q)$ are coordinate-dependent; their transformation law is given in~\cite{tv:microformal} and it is such that the canonical relation defined by an $S$ does not depend on a choice of coordinates. (Being a formal power series in $q$ is essential for that.) An example of a thick morphism and its generating function $S(x,q)$ is the thick morphism $\F$ specified by a function $S$ of the form $S=\f^i(x)q_i$. It corresponds to an ordinary map $\f\co M_1\to M_2$ and $\F$ is its lifting to the cotangent  bundles (so that $p_a=\der{\f^i}{x^a}(x)q_i$). (Note that $\F$ is   not a map.) Adding more terms in the expansion~\eqref{eq.S} gives a generalization of ordinary maps. The crucial construction is the following.

\begin{definition}[Pullback by a thick morphism]
For a thick morphism $\F\co M_1\tto M_2$ with a generating function $S=S(x,q)$, the \emph{pullback} of functions is defined by the formula:
\begin{equation}\label{eq.phistar}
     \boxed{\quad f(x)=g(y) + S(x,q) - y^iq_i\,, \quad\vphantom{\int_0^1} }
\end{equation}
where $q_i$ and $y^i$ are determined from the equations
\begin{equation}\label{eq.phistarq}
    q_i=\der{g}{y^i}\,(y)
\end{equation}
and
\begin{equation}\label{eq.phistary}
    y^i=(-1)^{\itt}\,\der{S}{q_i}(x,q)\,.
\end{equation}
So a function $g(y)$ on $M_2$ is mapped to a function $f(x)$ on $M_1$, and we write $f=\F^*[g]$.
\end{definition}

\begin{remark} Equations~\eqref{eq.phistarq} and~\eqref{eq.phistary} are coupled; if we substitute~\eqref{eq.phistarq} into~\eqref{eq.phistary}, we obtain
\begin{equation}\label{eq.phig}
    y^i=(-1)^{\itt}\,\der{S}{q_i}\Bigl(x,\der{g}{y}(y)\Bigr)
\end{equation}
which is then solved by iterations (see~\cite{tv:nonlinearpullback}) giving a formal perturbation $\f_g\co M_1\to M_2$ of the ordinary map $\f\co M_1\to M_2$ defined by the second term in~\eqref{eq.S} and depending on a function $g$ as a ``small parameter''\,:
\begin{equation}\label{eq.phigexpl}
    y^i=\f^i(x) + S^{ij}(x)\p_ig(\f(x)) +\ldots
\end{equation}
This is substituted into~\eqref{eq.phistarq} and then in~\eqref{eq.phistar}. This results in a nonlinear (in general) dependence of the function $f$ on a function $g$. Hence the pullback
\begin{equation}\label{eq.phistar1}
    \Phi^*\co \funn(M_2) \to \funn(M_1)\,,
 \end{equation}
is a nonlinear formal mapping of infinite-dimensional manifolds of functions. (In the supercase, one should care about parities. The above construction makes sense only for even functions or `bosonic fields'.)
\end{remark}

\begin{example} \label{ex.classordpull}
If $S=\f^i(x)q_i$, so $\F$ corresponds to an ordinary map $\f\co M_1\to M_2$, then one can check that formula~\eqref{eq.phistar} gives
\begin{equation}
    \F^*[g](x)=g\bigl(\f(x)\bigr)\,,
\end{equation}
i.e. $\F^*=\f^*$ is the ordinary pullback (in particular, linear). (Indeed, equations~\eqref{eq.phistarq} and~\eqref{eq.phistary} here decouple and give    $y^i=\f^i(x)$, so $\F^*[g](x)=g(y)+\f^i(x)q_i-y^iq_i=g(y)=g(\f(x))$.)
\end{example}
\begin{example}
For a general $S(x,q)$ with an expansion~\eqref{eq.S}, one obtains
\begin{equation}
    \Phi^*[g](x)= S^0(x)+ g\bigl(\f(x)\bigr) + \frac{1}{2}\,S^{ij}(x)\,\p_i g\bigl(\f(x)\bigr)\p_j g\bigl(\f(x)\bigr)
    + \ldots \,.
\end{equation}
with higher terms involving higher derivatives of $g$.
\end{example}

Although there is no known closed formula for pullbacks $\F^*$, it is possible from  the main construction~\eqref{eq.phistar} to obtain various remarkable properties~\cite{tv:nonlinearpullback, tv:microformal}. They include:  the formula for the \textbf{derivative} of the pullback:
\begin{equation}\label{eq.derivative}
    T\F^*=\f_g^*\,,
\end{equation}
where $\f_g\co M_1\to M_2$ is the perturbed map described above; and the \textbf{composition formula}
\begin{equation}\label{eq.functor}
    (\F_{32}\circ \F_{21})^*=\F_{21}^*\circ \F_{32}^*
\end{equation}
for thick morphisms $\F_{21}\co M_1\tto M_2$ and $\F_{32}\co M_2\tto M_3$, where the composition of the pullbacks $\F_{21}^*$ and $\F_{32}^*$  is the composition of (formal) mappings, while the \emph{composition of thick morphisms} is defined by the formula:
\begin{equation}\label{eq.compphi}
     \boxed{\quad S_{31}(x,r)=S_{32}(y,r) + S_{21}(x,q) - y^iq_i\,, \quad\vphantom{\int_0^1} }
\end{equation}
where $q_i$ and $y^i$ are determined from the equations
\begin{equation}\label{eq.compphiq}
    q_i=\der{S_{32}}{y^i}\,(y,r)
\end{equation}
and
\begin{equation}\label{eq.compphiy}
    y^i=(-1)^{\itt}\,\der{S_{21}}{q_i}(x,q)\,.
\end{equation}
Here  $S_{32}(y,r)$ and $S_{21}(x,q)$ are generating functions of $\F_{32}$ and $\F_{21}$, and $S_{31}(x,r)$ is by the definition the generating function of the composition  $\F_{32}\circ \F_{21}\co M_1\tto M_3$.\footnote{Formula~\eqref{eq.compphi} agrees with the usual set-theoretic composition of relations.} (The notation used is $x^a,p_a$ and $y^i,q_i$ as before, and $z^{\mu},r_{\mu}$ for position and momentum variables for $M_3$.)

\begin{remark} As already mentioned, there are parallel constructions  for odd functions. They   use the anticotangent bundles $\Pi T^*M$ instead of cotangent bundles and odd generating functions. (Above, the generating functions are even.)    ``Odd thick morphisms'' $\Psi\co M_1\oto M_2$ induce nonlinear pullbacks of odd functions or `fermionic fields' $\Psi^*\co \pfunn(M_2)\to \pfunn(M_1)$. They have many properties similar to the thick morphisms described above (which may be branded ``even''), except for one: they are lacking a ``quantum'' counterpart (see below).
\end{remark}

\subsubsection{Quantum thick morphisms}
It turns out that there is a certain ``quantum version'' of (even) thick morphisms  $\F\co M_1\tto M_2$. To distinguish, we shall call the latter ``classical''. ``Quantum thick morphisms'' are defined via the corresponding pullbacks, which are introduced first  and    quantum thick morphisms as such are then  arrows in the dual category.

As for classical thick morphisms, quantum thick morphisms are specified by their generating functions. In local coordinates they look the same, $S=S(x,q)$, but now they may depend on Planck's constant $\hbar$ (as a formal power series) and they will have a different transformation law under a change of coordinates. For   distinction, we refer to them as \emph{quantum generating functions}.

Functions on which pullbacks by quantum thick morphisms will act are \emph{oscillatory wave functions}, i.e. linear combinations of formal exponentials such as $w(x)=\sum A(x)e^{\frac{i}{\hbar}f(x)}$ where both the phase $f(x)$ and the amplitude $A(x)$ are formal power series in $\hbar$ (with nonnegative powers only) and natural rules of manipulation with such expressions are assumed. Coefficients of these power expansions in $\hbar$ are smooth functions on our manifolds (unless otherwise is stated).

\begin{definition}[Pullback by a quantum thick morphism]
The \emph{pullback} $\hat\F^*$ by a \emph{quantum thick morphism} $\hat \F\co M_1\ttoh M_2$ specified by a quantum generating function $S(x,q)$ is the following integral operator
\begin{equation}\label{eq.phihat}
    \boxed{\quad (\hat\F^* w)(x)= \int_{T^*M_2} \Dbar(y,q) \,\, e^{\frac{i}{\hbar}\left(S(x,q)-y^iq_i\right)}\,w(y)\,,\quad\vphantom{\int\limits_0^1}}
\end{equation}
mapping oscillatory wave functions $w(y)$ on $M_2$ to oscillatory wave functions $u(x)$  on $M_1$.
Integration in~\eqref{eq.phihat} is with respect to the normalized Liouville measure  on $T^*M_2$, so that  the notation $\Dbar(y,q)=Dy \Dbar q$  means that no numerical factors would appear from the $\hbar$-Fourier transform.
\end{definition}

We  stress that   an integral operator $\hat\F^*$ is   a primary object,  but it is  considered   as corresponding   to an arrow $\hat \F\co M_1\ttoh M_2$   in the opposite direction introduced formally.

The transformation law for a quantum generating function $S(x,q)$ is such that the integral~\eqref{eq.phihat} is invariant under a  change  of coordinates.  One can see that it is different from the transformation law for classical generating functions (but the difference is of order $\hbar$).

\begin{example}
Consider $S(x,q)=\f^i(x)q_i$. Then
\begin{equation}\label{eq.ordpullback}
    (\hat\F^* w)(x)= \int_{T^*M_2} \Dbar(y,q) \,\, e^{\frac{i}{\hbar}\left((\f^i(x) -y^i)q_i\right)}\,w(y)=
    \int_{M_2}  Dy \,\delta(\f(x)-y)\,w(y) = w(\f(x))\,,
\end{equation}
so we have the ordinary pullback by a usual map $\f\co M_1\to M_2$. (Compare the same for classical thick morphisms in Example~\ref{ex.classordpull}.)
\end{example}

The general form of a quantum pullback can be seen from the following proposition. We can write
\begin{equation}\label{eq.splus}
    S(x,q)=S^0(x)+\f^i(x)q_i +S^{+}(x,q)\,,
\end{equation}
where $S^{+}(x,q)$ contains all terms of order $\geq 2$ in $q$.
\begin{proposition} The operator $\hat \F^*$ corresponding to $S(x,q)$ as in~\eqref{eq.splus} can be written explicitly as
\begin{equation}\label{eq.Phiexplic}
     \bigl(\hat \Phi^*w\bigr)(x)=e^{\frac{i}{\hbar}S^0(x)}
    \left(e^{\frac{i}{\hbar}S^{+}\left(x,\frac{\hbar}{i}\der{}{y}\right)}w(y)\right)_{\left|\vphantom{\int\limits_a^b}\ y^i=\f^i(x)\right.}\,.
\end{equation}
\end{proposition}
(This is a formal $\hbar$-differential operator of infinite order along a map $\f\co M_1\to M_2$.)

The main properties of quantum thick morphisms are the following. Every quantum thick morphism $\hat \F$ with a generating function $S(x,q)$ has a  \textbf{classical limit}  $\F=\lim\limits_{\hbar\to 0}\hat \F$, which is the (classical) thick morphism with the generating function $S^{\text{class}}(x,q)$ obtained by setting $\hbar=0$ in $S(x,q)$, so that for an oscillatory wave function of the form $w(y)=e^{\frac{i}{\hbar}g(y)}$,
\begin{equation}\label{eq.classlim}
    \hat\F^*[e^{\frac{i}{\hbar}g(y)}](x) = e^{\frac{i}{\hbar}f(x)}\,,
\end{equation}
where
\begin{equation}\label{eq.classlimf}
    f(x)=\F^*[g](x)\,(1+O(\hbar))
\end{equation}
(with the pullback by the classical thick morphism in the right-hand side). Similarly to the classical case, holds the formula $(\hat\F_{32}\circ \hat\F_{21})^*=\hat\F_{21}^*\circ \hat\F_{32}^*$, where the composition of two quantum thick morphisms with generating functions $S_{32}(y,r)$ and $S_{21}(x,q)$ is the quantum thick morphism  with the generating function $S_{31}(x,r)$ defined by the equation
\begin{equation}\label{eq.qucomp}
    \boxed{\quad e^{\frac{i}{\hbar}S_{31}(x,r)}=
    \int_{T^*M_2} \Dbar(y,q) \,\,
    e^{\frac{i}{\hbar}\left(S_{32}(y,r) + S_{21}(x,q) - y^iq_i\right)}\,.\quad\vphantom{\int\limits_0^1}}
\end{equation}
By the stationary phase formula (see the Appendix in~\cite{tv:microformal}), one can obtain that the ``classical'' composition~\eqref{eq.compphi} is the limit of the ``quantum'' composition~\eqref{eq.qucomp}.

\begin{remark}
Integrals similar to~\eqref{eq.phihat} appeared first in the context of quantum mechanics in the work by Fock~\cite{fock:canon1959-69}  and independently in PDE theory in the works of several people, notably Egorov~\cite{egorov:canonicpdo1971} and Fedoryuk~\cite{fedoryuk:pdo1971}, that paved way for H\"{o}rmander's general theory of Fourier integral operators~\cite{hoer:fio1-1971}. Our  pullbacks by quantum thick morphisms can be regarded as   (special type) ``$\hbar$-Fourier integral operators''.
\end{remark}

\section{Infinitesimal calculus for thick morphisms}\label{sec.hamjac}

We   concentrate here on \emph{even} thick morphisms, though some statements carry over to odd thick morphisms as well.

\subsection{Classical thick diffeomorphisms and Hamilton--Jacobi equation}

Let $M$ be a manifold or supermanifold (the distinction plays no role here). We shall refer to a thick morphism $\F\co M\tto M$ as   a \emph{thick diffeomorphism} if $\F$ is invertible.

We consider the following situation. Let $\F_t\co M\tto M$, where $t\in\RR$, be a $1$-parameter group of thick diffeomorphisms, i.e. $\F_{t+s}=\F_t\circ \F_s$, $\F_0$ is the identity and $\F_{-t}=(\F_t)^{-1}$. As in the usual case,   thick diffeomorphisms $\F_t$ may be defined only on some open $U\subset M$, but this would make no difference for our analysis. Let $S_t=S_t(x,q)$ be the generating function of $\F_t$ in some coordinates on $M$. (Note that $x^a$ are local coordinates of a point $x$ of $M$, while $q_a$ are the components of a momentum at some other point $y$ of $M$.) We want to study the evolution of $S_t$ as well as that of the pullback by $\F_t$ of a function $f(x)$.

Consider first an infinitesimal thick diffeomorphism $\F_{\e}$, $\e^2=0$. Since $\F_0=\id$, we have
\begin{equation}\label{eq.s0}
    S_0(x,q)=x^aq_a\,,
\end{equation}
the generating function of the identity map. Here $q_a$ are momenta at the same point $x=y$. For an infinitesimal thick diffeomorphism we therefore have
\begin{equation}\label{eq.seps}
    S_{\e}(x,q)=x^aq_a+\e H(x,q)\,.
\end{equation}
One can observe that here $q_a$ are  momenta at a point $y$ which is in the infinitesimal ``$\e$-neighborhood'' of $x$, so in the argument of $H(x,q)$ because of the presence of the factor $\e$ they can be seen as being   at the original point $x$. Hence $H=H(x,p)$ is a genuine function on $T^*M$, i.e. a Hamiltonian on $M$.
\begin{lemma}\label{lem.infinites}
For a function $f\in \fun(M)$, its infinitesimal pullback by $\F_{\e}$ is given by
\begin{equation}\label{eq.phistareps}
    \F^*_{\e}[f](x)=f(x)+\e H\bigl(x,\der{f}{x}\bigr)\,.
\end{equation}
\end{lemma}
(Note that $f$ in the Lemma must be even, if we consider the super case.)
\begin{proof}   We have, by the definition,
\begin{equation*}
    \F^*_{\e}[f](x)= f(y)+ x^aq_a+\e H(x,q) -y^aq_a\,,
\end{equation*}
where
\begin{equation*}
    y^a=(-1)^{\at}\der{S}{q_a}(x,q)=x^a+\e \der{H}{p}(x,q)\,, \quad q_a=\der{f}{x^a}(y)\,,
\end{equation*}
so in particular $y^a-x^a\sim \e$. Hence in the formula we have
\begin{equation*}
    f(y)= f\bigl(x+\e \der{H}{p}(x,q)\bigr)=f(x) +\e \der{H}{p_a}\bigl(x,q\bigr)q_a=f(x) +\e (y^a-x^a)q_a\,,
\end{equation*}
so
\begin{equation*}
    \F^*_{\e}[f](x)= f(x) +\e (y^a-x^a)q_a+ x^aq_a+\e H(x,q) -y^aq_a=f(x) +\e H\bigl(x,\der{f}{x}(x)\bigr)\,,
\end{equation*}
as claimed.
\end{proof}

We   refer to the Hamiltonian $H\in\fun(T^*M)$ as the \emph{generator} of a  $1$-parameter group of thick diffeomorphisms $\F_t$. Since a Hamiltonian such as $H$ generates also a Hamiltonian flow, i.e. a $1$-parameter group of canonical transformations $F_t\co T^*M\to T^*M$, we   need to clarify a relation between $\F_t$ and $F_t$. This is done further below (see Proposition~\ref{prop.Stasintegral} and the discussion around it).

Lemma~\ref{lem.infinites} leads to a description of a  $1$-parameter group of thick diffeomorphisms $\F_t$ with a generator $H$ in terms of differential equations. They are   Hamilton--Jacobi equations  with the Hamiltonian $H$.

\begin{theorem} \label{thm.hamjacft}
Let $f_t:=\F_t^*[f_0]$ for some  initial  function $f_0=f(x)$. Then $f_t=f_t(x)$ satisfies the differential equation
\begin{equation}\label{eq.derft}
    \der{f_t}{t}=H\Bigl(x,\der{f_t}{x}\,\Bigr)\,.
\end{equation}
\end{theorem}
\begin{proof}
We have $f_{t+\e}=\F^*_{t+\e}[f_0]=(\F_{t}\circ\F_{\e})^*[f_0]=\F^*_{\e}(\F^*_t[f_0])=\F_{\e}^*[f_t]=f_t+\e H\bigl(x,\der{f_t}{x}(x)\bigr)$, which implies the statement.
\end{proof}

\begin{remark}
The statements of Lemma~\ref{lem.infinites} and Theorem~\ref{thm.hamjacft} were first observed in~\cite{tv:oscil} and they led  there  to introduction of quantum thick morphisms (prompted by the relation of  Hamilton--Jacobi equation and quantum mechanics).
\end{remark}

\begin{theorem} The generating function $S_t=S_t(x,q)$ satisfies
\begin{equation}\label{eq.derst}
    \der{S_t}{t}=H\Bigl((-1)^{\qt}\der{S_t}{q},q\Bigr)\,.
\end{equation}
\end{theorem}
(The sign in the argument is, in greater detail, $(-1)^{\at}$ for the partial derivative in $q_a$.)
\begin{proof}
Consider $\F_{t+\e}= \F_{\e}\circ  \F_t$ and write the corresponding composition formula for generating functions. We have to consider three copies of $M$,
\begin{equation*}
    \underset{x^a,p_a}{M} \ttto{\Phi_{t}} \underset{y^a,q_a}{M} \ttto{\Phi_{\e}} \underset{z^a,r_a}{M}\,,
\end{equation*}
where we have listed, for clarity, our notations for the positions and momenta. By the composition formula,
\begin{equation*}
    S_{t+\e}(x,r)=S_t(x,q)+S_{\e}(y,r) -y^aq_a\,,
\end{equation*}
where we can   use formula~\eqref{eq.seps} for $S_{\e}$. Hence
\begin{align*}
    y^a&=(-1)^{\at}\der{S_t}{q_a}(x,q)\,,\\
    q_a&=\der{S_{\e}}{y^a}(y,r)=r_a+\e (-1)^{\tilde\e}\der{H}{x^a}(y,r)= r_a+\e (-1)^{\tilde\e}\der{H}{x^a}(y,q)\,,
\end{align*}
where in the latter formula we could replace $r$ by $q$ in the argument because $r-q\sim \e$. Therefore
\begin{multline*}
    S_{t+\e}(x,r)=S_t(x,q)+y^ar_a+\e H(y,r)-y^aq_a=\underbrace{S_{t}(x, r+(q-r))}_{S_t(x,r)+(q-r)_a\der{S_t}{q_a}(x,r)}-y^a(q_a-r_a)+\e H(y,r)=\\
    S_t(x,r)+(q-r)_a\der{S_t}{q_a}(x,r)-(-1)^{\at}(q_a-r_a)y^{a} +\e H(y,r)\,.
\end{multline*}
where $y^{a}=(-1)^{\at}\der{S}{q_a}(x,q)$. Observe that in this formula we can replace $q$ by $r$ in the arguments (because their difference is of order $\e$ and the coefficients are also of order $\e$). Then the middle terms cancel and we arrive finally at
\begin{equation*}
    S_{t+\e}(x,r)=S_t(x,r)+ \e H\left((-1)^{\at}\der{S_t}{q_a}(x,q),r\right)
\end{equation*}
(where we can return to $q$ instead of $r$, for aesthetical purposes). This gives the differential equation~\eqref{eq.derst}.
\end{proof}

Now we shall clarify the relation of generating functions for thick morphisms with the ``textbook'' notion of action in classical mechanics (as in e.g.~\cite[\S 43]{landau-lifshitz:mech}).

Let $y=y(t), q=q(t)$ be the Hamiltonian flow with a Hamiltonian $H$.
Reserve $x,p$ for the initial conditions and use $y,q$ for the dynamic variables of position and momentum. Consider   the classical action of this system, which is the integral of the $1$-form $q_ady^a-H(y,q)dt$ taken over the `true' trajectory, i.e. satisfying Hamilton's equations. Denote it   $W=W_t$. Here we   fix the initial time $t=0$ and vary the final time  $t$, as well as the initial and final position and momentum, which we denote $x,p$ and $y,q$ (note once again that, contrary to the traditional notation, $q$ in our notation is a momentum, not a position). By  standard argument,
\begin{equation}\label{eq.dF}
    dW_t= dy^aq_a-dx^ap_a - H(y,q)dt\,.
\end{equation}
Now we can take the Legendre transform from $y$ to $q$ and define $S_t=y^aq_a-W_t$. 
Then we have
\begin{equation}\label{eq.dS}
    dS_t= dy^aq_a+dx^ap_a + H(y,q)dt\,,
\end{equation}
or
\begin{equation}\label{eq.eqforS}
    \left\{\ \begin{aligned}
    \der{S_t}{t}(y,q) &=H\Bigl((-1)^{\qt}\der{S_t}{q}(y,q),q\Bigr)\\
    \der{S_t}{q_a}(y,q) &=(-1)^{\at}y^a\,, \quad
    \der{S_t}{x^a}(y,q) =p_a\,.
    \end{aligned}\right.
\end{equation}
Hence we see that $S_t=S_t(x,q)$ is the generating function of the Lagrangian submanifold in $T^*M\times T^*M$ which is the graph of the canonical transformation $(x,p)\mapsto (y,q)$ at   time $t$. This is of course classical. But $S_t=S_t(x,q)$ is also the generating function of the corresponding thick diffeomorphism $\F_t\co M\tto M$. Hence we can summarize as follows.
\begin{proposition}\label{prop.Stasintegral}
The generating function $S_t(x,q)$
of the $1$-parameter family of thick diffeomorphisms with a Hamiltonian $H=H(x,p)$
is the Legendre transform of the classical action function
$W=W_t(x,y)$ for the corresponding Hamiltonian system. Therefore
\begin{equation}\label{eq.Sasint}
    S_t(x,q)= x^aq_a + \int_0^t\bigl((-1)^{\at}y^adq_a + H(y(t),q(t))dt\bigr)\,
\end{equation}
(integral over  the trajectory of the Hamiltonian flow).
\qed
\end{proposition}

In the considered case, the generating function $S_t(x,q)$ itself can be referred to as ``action'' following the traditional practice when quantities related by Legendre transform  are called by the same name (e.g. various types of energy in thermodynamics).
This justifies     calling, in  general, a generating function $S(x,q)$  of an arbitrary  thick morphism $\F\co M_1\tto M_2$ the \emph{action} of $\F$ (which we shall now do    interchangeably with  `generating function').

\subsection{Quantum thick diffeomorphisms and the Schr\"{o}dinger equation}

Consider a $1$-parameter group of \emph{quantum thick diffeomorphisms} $\hat\F_t$. They correspond to a $1$-parameter group of integral operators $\hat\F_t^*$ acting on oscillatory wave functions $w(x)$ by
\begin{equation}\label{eq.quantdiffeo}
    (\hat\F^*_tw)(x)=\int_{T^*M} \Dbar(y,q) e^{\frac{i}{\hbar}(S_t(x,q)-yq)} \,w(y)\,.
\end{equation}
(We remind that notation such as $\Dbar(y,q)$ means normalized Liouville measure, i.e. subsuming   factors like $2\pi\hbar$.)
Here $S_t=S_t(x,q)$ is a formal power series in $\hbar$ (which we do not indicate explicitly) and we may refer to it as a ``quantum action''. 
As a geometric object, it differs from the ``classical'' action of a (classical) thick morphism considered in the previous subsection. Nevertheless, it is clear that for the identity morphism, the action is still $S=x^aq_a$, since
\begin{equation*}
    \int \Dbar(y,q) e^{\frac{i}{\hbar}(xq -yq)} w(y)=  \int \Dbar(y,q) e^{\frac{i}{\hbar}(x -y)q} \,w(y)=
    \int  D y\, \delta(x -y) w(y)= w(x)\,
\end{equation*}
(which is a particular case of \eqref{eq.ordpullback}). For an infinitesimal quantum diffeomorphism,
\begin{equation}\label{eq.qseps}
    S_{\e}(x,q)=x^aq_a +\e H^{\hbar}(x,q)\,,
\end{equation}
where we have emphasized possible dependence on $\hbar$ in $H^{\hbar}(x,q)$.  The function $H^{\hbar}(x,p)$ can be seen as the quantum analog of a (classical) Hamiltonian $H(x,p)$ appearing in equation~\eqref{eq.seps}. We shall refer to $H^{\hbar}(x,p)$ as the \emph{generator} of an infinitesimal quantum thick diffeomorphism.  Its geometric nature will become  clearer when we write down the corresponding infinitesimal pullback.

\begin{lemma}
\label{lem.qphieps}
For an oscillatory wave function $w$ on $M$, its pullback by an infinitesimal quantum thick diffeomorphism with a generator $H^{\hbar}(x,p)$  is given by
\begin{equation}\label{eq.qphieps}
    \hat\F^*_{\e}[w] (x) = w(x) +\e\, \frac{i}{\hbar}   \Hqu w(x)\,.
\end{equation}
Here in $\Hqu$   the indices $1$ and $2$ atop of the non-commuting operators substituted for the   arguments 
indicate the ordering where all   $x$'s stand to the left of all $\frac{\hbar}{i}\der{}{x}$'s.
\end{lemma}
\begin{proof}
For the ``infinitesimal'' quantum action $S_{\e}(x,q)$ given by~\eqref{eq.qseps}, we have the quantum pullback $\hat\F^*_{\e}$, and
\begin{multline*}
    \hat\F^*_{\e}w(x)= \int_{T^*M}\Dbar(y,q)\underbrace{e^{\frac{i}{\hbar}(x^aq_a+\e H^{\hbar}(x,q)-y^aq_a)}}_{\bigl(1+\e \frac{i}{\hbar}H^{\hbar}(x,q)\bigr)e^{\frac{i}{\hbar}(x^a-y^a)q_a}}w(y)=
    w(x)+\e \frac{i}{\hbar}H^{\hbar}\left(x,\frac{\hbar}{i}\der{}{y}\right) w(y)_{\left|y=x\right.}\,,
\end{multline*}
as claimed.
\end{proof}

\begin{remark}
We see from the lemma that $H^{\hbar}(x,p)$ is the full symbol of a quantum Hamiltonian $\hat H$ as based on $xp$-quantization. (Which is close to standard  coordinate-dependent full symbol  in the theory of pseudodifferential operators.) So the  transformation law of $H^{\hbar}(x,p)$ as    a geometric object  is different from that for   classical Hamiltonians, i.e.   genuine functions on $T^*M$. Looking at the expansion in $\hbar$, we can see that in the zeroth order in $\hbar$, the function $H^{\hbar}(x,p)$ transforms as a genuine function on $T^*M$ and then there are ``quantum corrections'' in the transformation law.
\end{remark}

\begin{theorem}
Suppose $w=w_t(x)$ is obtained from some initial $w_0(x)$ by the pullback by a $1$-parameter group of quantum diffeomorphisms  $\hat\F_t$ with a generator $H^{\hbar}=H^{\hbar}(x,p)$. Then the function $w$ satisfies the non-stationary Schr\"{o}dinger equation
\begin{equation}\label{eq.schr}
    \frac{\hbar}{i}\der{w}{t} = \hat H w\,,
\end{equation}
where $\hat H = \Hqu$.
\end{theorem}
\begin{proof}
We have $w_{t+\e}=\hat\F^*_{t+\e}[w_0]=\hat\F^*_{\e}(\hat\F^*_t[w_0])=\hat\F_{\e}^*[w_t]=w_t+   \e\, \frac{i}{\hbar}\Hqu w(x)$.
\end{proof}

\subsection{Main derivation formula for thick morphisms}

Now we move from $1$-parameter groups of (classical or quantum) thick diffeomorphisms  of a single manifold $M$ to the general case of thick morphisms $\F\co M_1\tto M_2$ between possibly different manifolds. Suppose there is a $1$-parameter family of thick morphisms $\F_t\co M_1\tto M_2$. Instead of a generator that we had for $1$-parameter groups of   thick diffeomorphisms, we now have a ``velocity'', which is a time-dependent ``Hamiltonian'' $H_t(x,q)$ (it is not a true Hamiltonian). We can consider the quantum case $\hat\F_t\co M_1\ttoh M_2$ as well. Since  there are no $1$-parameter groups now, we do not obtain differential equations, just derivation formulas for   pullbacks in terms of the original functions, which we shall also assume depending on parameter $t$.

Consider a (classical or quantum) action $S_t=S_t(x,q)$ and define
\begin{equation}\label{eq.ht}
    H_t(x,q)=\der{S_t}{t}(x,q)\,.
\end{equation}

Then we have the following statements.

\begin{theorem}
\label{thm.derqua}
For a quantum pullback  of an oscillatory wave function  $w_t(y)$,
\begin{equation}
\frac{\hbar}{i}\oder{}{t}\hat\F_t^*[w_t]=   (\hat H_t\hat F^*)(w_t) +\frac{\hbar}{i}\,\hat\F^*\Bigl[\der{w_t}{t}\Bigr]\,,
\end{equation}
where the operator $\hat H_t\hat F^*$ taking oscillatory wave functions on $M_2$ to oscillatory wave functions on $M_1$ is defined as
\begin{equation}\label{eq.hphit}
    \hat H_t\hat F^*=
     e^{\frac{i}{\hbar}S^0(x)}
    \left(H_t\Bigl(x,\frac{\hbar}{i}\der{}{y}\Bigr)e^{\frac{i}{\hbar}S^{+}\left(x,\frac{\hbar}{i}\der{}{y}\right)} \right)_{\left|\vphantom{\int\limits_a^b}\ y =\f(x)\right.}\,.
\end{equation}
\end{theorem}

(Here we use the expansion~\eqref{eq.splus}.)

\begin{theorem} \label{thm.derclass}
For a classical pullback  of a function $g_t(y)$,
\begin{equation}
\oder{}{t}\F_t^*[g_t](x)=   H_t\Bigl(x,\f^*_{g_t}\Bigl(\der{g_t}{y}\Bigr)\Bigr)+\f_{g_t}^*\Bigl(\der{g_t}{t}\Bigr)\,.
\end{equation}
\end{theorem}

We suppress  proofs of Theorems~\ref{thm.derqua} and~\ref{thm.derclass}. Note that  Theorem~\ref{thm.derclass} in particular contains  formula~\eqref{eq.derivative} for the derivative of the non-linear pullback by a (classical) thick morphism.

\section{Quantum thick morphisms and  spinor representation}\label{sec.spin}

\subsection{Two words about quantization and spinors}

We would like to recall some   facts about the spinor representation in the form that will help us to establish a relation with classical and quantum thick morphisms (our goal). It will be embracing both ``orthogonal spinors'' and ``symplectic spinors''. (``Symplectic'' version of spinor representation is also known as metaplectic representation or Shale--Weil representation, and by other names.) In the orthogonal case, the construction is based on the Clifford algebra (associated with a given quadratic form), and in the symplectic case, with the Weyl algebra. In fact, these two setups become indistinguishable if we work with super vector spaces since parity reversion turns an orthogonal structure into symplectic and vice versa, and the same for Clifford and Weyl algebras. For concreteness   we  speak here about Weyl algebra (in the supercase, but we  do not stress it).

We   view the Weyl algebra as an associative algebra over $\C$ generated by elements $\hat p_a$, $\hat x^b$ and $\hbar$ satisfying the Heisenberg relation
\begin{equation}\label{eq.heisrel}
    [\hat p_a, \hat x^b]=\frac{\hbar}{i}\,\delta_a^b\,,
\end{equation}
where the square bracket means commutator, and all   other commutators between the generators vanish  (in particular, the element $\hbar$ is central). Clearly, the quotient by the ideal generated by $\hbar$ (i.e. ``setting Planck's constant to zero'') is the commutative algebra of polynomials in   $p_a$ and $x^b$, the images of $\hat p_a$, $\hat x^b$.  Since in the Weyl algebra the commutator of arbitrary elements  is divisible by $\frac{\hbar}{i}$, we can define the \emph{quantum Poisson bracket} on the Weyl algebra by
\begin{equation}\label{eq.qpoiss}
    \{\hat f, \hat g\}_{\hbar}:= \frac{i}{\hbar}\, [\hat f, \hat g]
\end{equation}
 (Dirac's definition) and the \emph{classical Poisson bracket} on functions in $p$, $x$ by
\begin{equation}\label{eq.classandquantbrack}
    \{f,g\}:= \{\hat f,\hat g\}_{\hbar} \mod \hbar\,,
\end{equation}
where $f=\hat f\mod \hbar$ and $g=\hat g\mod \hbar$ (so $\hat f$ and $\hat g$ are arbitrary liftings of $f$ and $g$ to the Weyl algebra). The latter is well-defined and is of course the usual Poisson bracket. We can refer to elements of the Weyl algebra as ``quantum Hamiltonians'' and to the elements of the polynomial algebra as ``classical Hamiltonians''. The natural projection that sends a quantum Hamiltonian $\hat f$ to the classical Hamiltonian $f=\hat f \mod \hbar$ is the \emph{principal symbol} map (it is exactly the $\hbar$-principal symbol if elements of the Weyl algebra are  realized as $\hbar$-differential operators). So in particular, the principal symbol map is a Lie algebra homomorphism with respect to the Poisson brackets. There is no, of course,  natural map in the opposite direction, i.e. that sends a classical Hamiltonian $f$ to a quantum Hamiltonian $\hat f$ for which $f$ is the principal symbol, $f=\hat f \mod \hbar$. Each such a map is a ``quantization'', and there are many quantizations. They  in general do not preserve Poisson brackets. (Since the beginning of quantum mechanics it is well known that a choice of quantization is basically a choice of ordering of $\hat p_a$ and $\hat x^b$.) By ``extending scalars'' for classical Hamiltonians so make it possible for them to depend  on $\hbar$, one can make a quantization into a linear bijection between classical and quantum Hamiltonians, the inverse map being a  \emph{full symbol map} (non-canonical, as opposed to the principal symbol map). Through the non-uniqueness of quantization, classical Hamiltonians can receive ``quantum corrections'' by applying first a quantization map $Q_1$ and then an inverse quantization (or full symbol) map $Q_2^{-1}$, $f\mapsto Q_2^{-1}Q_1 (f)=f + \hbar(...)$. (In this way, ``quantized classical Hamiltonians'' can change their behavior as geometrical objects. An example is the transformation law for full symbols of $\hbar$-differential or $\hbar$-pseudodifferential operators, which we already mentioned.)

How all that applies to spinor representation?

In the Weyl algebra consider a linear subspace spanned by $\hat p_a$ and $\hat x^b$. Denote it $L$. We shall refer to its elements as   linear (quantum) Hamiltonians. (In the Clifford algebra one takes the subspace spanned by $\hat\g^{\mu}$.) Inside the   group of all invertible elements of the Weyl algebra (or its suitable completion making   possible to consider e.g. exponentials) one can distinguish a closed subgroup $G$ specified by the condition that
\begin{equation}\label{eq.gLginv}
    \hat gL\hat g^{-1}\subset L
\end{equation}
for all $\hat g\in G$. The group $G$ is essentially the (symplectic) \emph{spinor group}. (In traditional usage, this name may be reserved to the orthogonal version.) By construction, it is defined together with a group homomorphism $\hat g\mapsto T_{\hat g}$ to the group of linear transformations of the vector space $L$, $T_{\hat g}(\hat a)=\hat g\hat a\hat g^{-1}$. Since the adjoint action  preserves   commutation relations, it follows that  $T_{\hat g}$ takes values in the group of linear symplectic (=linear canonical) transformations. (By expanding the ``new'' canonical variables $T_{\hat g}(\hat p_a)$, $T_{\hat g}(\hat x^a)$ over $\hat p_a$, $\hat x^a$ and taking the coefficients, we can obtain a realization of $T_{\hat g}$ as a symplectic matrix.)

The \emph{spinor representation} of the symplectic or orthogonal group is the inverse map $T_{\hat g}\mapsto \hat g$ combined with a (unique up to equivalence) realization of Weyl or Clifford algebra by linear operators acting on the space of   \emph{spinors}\,\footnote{E.g. the standard realization $\hat x^a=x^a$ and $\hat p_a=\frac{\hbar}{i}\lder{}{x^a}$ on functions of coordinates   or  alternatively by using creation-annihilation operators and the holomorphic realization in the Bargmann--Fock space. For Clifford algebra one needs to use differential operators with odd variables (see~\cite{berezin:marinov2}, also~\cite{tv:as}),   particular options depending on the signature of the quadratic form.}. It is multi-valued because $T_{\hat g}\mapsto \hat g$ is multi-valued.  Extra normalization conditions may be imposed on elements of $G$ to reduce the kernel of the group  homomorphism $\hat g\mapsto T_{\hat g}$ and thus the multi-valuedness of the inverse map. To see the nature of possible conditions, it is convenient to consider the infinitesimal case. Infinitesimal elements of the group $G$ have the form $\hat g=1+\e \hat H$ and the condition~\eqref{eq.gLginv} becomes
\begin{equation}\label{eq.gLginvinf}
    [\hat H, L]\subset L\,.
\end{equation}
This can hold  only for the quantum Hamiltonians $\hat H$ of degree $\leq 2$ in $\hat p_a$ and $\hat q^a$. They are referred to as ``quadratic'', but they contain also linear and constant terms. While linear Hamiltonians make invariant sense,  one cannot canonically separate ``strictly quadratic'' Hamiltonians from scalars because a change of order results in adding a constant. It is known that there are different recipes for ordering: when all $\hat x^a$'s are to the left of $\hat p_b$'s, the other way round, the symmetric or Weyl ordering, and actually it is possible to interpolate these cases with a parameter $s\in[0,1]$ (see \cite{berezin:ishubin}). It is possible to express all these choices using integrals. Consider classical quadratic Hamiltonians, i.e. linear combinations of $p_ap_b$, $p_ax^b$ and $x^ax^b$. They make the Lie algebra of the symplectic group (= the group of linear canonical transformations). Consider the quantization map $Q_s$ based on the ordering with parameter $s$, so that
\begin{equation*}
    Q_s(p_ax^b)=s\,\hat p_a\hat x^b +(1-s)\,\hat x^b \hat p_a= \hat x^b \hat p_a +s\,\frac{\hbar}{i}\delta_a^b\,.
\end{equation*}
By a direct check we can get the following well-known statement:
\begin{proposition}
For quadratic classical Hamiltonians $H_1$ and $H_2$,
\begin{equation}\label{eq.cocyc}
    \{Q_s(H_1),Q_s(H_2)\}_{\hbar}= Q_s\left(\{H_1,H_2\}\right) + (1-2s)\,\frac{\hbar}{i}\, c(H_1,H_2)\,,
\end{equation}
where $c(H_1,H_2)$ is a certain $2$-cocycle.
\end{proposition}
(If we write the coefficients of quadratic Hamiltonians in the matrix form, then the cocycle $c(H_1,H_2)$ will be expressed via   (super)trace and   matrix commutator.)

Hence only the Weyl ordering ($s=1/2$) does not lead to an extra term and the Weyl-ordered quadratic quantum Hamiltonians form a Lie algebra (under the quantum Poisson bracket) isomorphic to the symplectic algebra of classical quadratic Hamiltonians.

If   Weyl ordering is used to get rid of the ambiguity,  then the spinor representation can be made  isomorphism on the level of Lie algebras.  It will still be two-valued on the level of groups (as a section of  a non-trivial double cover). That is how the spinor representation is usually presented. The aim of   discussion here is to stress dependence on a choice of ordering that in general leads to appearing of a cocycle, i.e. the spinor representation becoming projective. (In the finite-dimensional case, the cocycle is coboundary because of the existence of the Weyl ordering. This is not  the case in infinite dimensions as found by Berezin~\cite{berezin:second}. It is argued  in Vershik\cite{vershik:metagonal1983}  that it is more natural to consider  the  spinor representation as projective  and respectively view the spinor group as the central extension\,---\,not a double cover\,---\, of the symplectic or orthogonal group   in all cases.)

By combining the above analysis with Lemma~\ref{lem.qphieps}, we immediately conclude the following.
\begin{corollary}
For a vector space $V$, infinitesimal quantum thick diffeomorphisms $\hat \F_{\e}\co V\ttoh V$ with quadratic generators $H^{\hbar}(x,q)$ give a projective version of  the spinor representation of the symplectic Lie superalgebra $\mathfrak{spo}(W)$, where $W=V\oplus V^*\cong T^*V$.
\end{corollary}
\begin{proof} Indeed,    the $xp$-ordering (s=0) has the cocycle $\frac{\hbar}{i}\, c(H_1,H_2)$.
\end{proof}

As for non-infinitesimal case, it is convenient to consider it in the general setting of thick morphisms between different vector spaces, which we shall do in the next subsection.

\begin{remark}  Why there is no analog of spinor representation for non-linear canonical transformations? One may notice that if the condition  $\hat gL\hat g^{-1}\subset L$ is dropped, the commutation relations are still preserved and the formulas $(x^a,p_a) \mapsto (\hat g\hat x^a \hat g^{-1}, \hat g\hat p_a \hat g^{-1})  \mod \hbar$ define  a non-linear canonical transformation in the space with coordinates $x,p$. It is a  natural (not depending on any choices)  homomorphism  of groups (as well as in the infinitesimal version, of Lie algebras) generalizing that for the linear case. What will be missing, is an inverse  quantization  map   which would preserve the Lie brackets. As   discussed, the problem exists already in the linear case, but there the ambiguity is only in a scalar term, which is somewhat masked by a choice of Weyl ordering. In the non-linear case, the obstruction is no longer  a scalar.
\end{remark}

\subsection{Quantum pullbacks as spinor representation}

The basis of our analysis will be the intertwining relation
\begin{equation}\label{eq.intertwin}
    \Delta_1\circ \hat\F^* = \hat\F^*\circ \Delta_2\,.
\end{equation}
Here $\hat\F\co M_1\ttoh M_2$ is a quantum thick morphism and $\Delta_1$ and $\Delta_2$ are $\hbar$-differential operators acting on oscillatory wave functions on $M_1$ and $M_2$ (see~\cite{tv:microformal}). In~\cite{tv:microformal}, it was applied to obtaining $\Linf$-morphisms of quantum Batalin--Vilkovisky algebras generated by $\Delta_1$ and $\Delta_2$ regarded as BV-operators (and $\Sinf$-algebras which are their classical limits). We shall show that the same intertwining relation leads to a version of spinor representation.

Together with~\eqref{eq.intertwin} we shall consider its classical limit
\begin{equation}\label{eq.hamjacob}
    H_1\Bigl(x,\der{S}{x}(x,q)\Bigr)=H_2\Bigl((-1)^{\qt}\der{S}{q}(x,q),q\Bigr)
\end{equation}
(see~\cite{tv:nonlinearpullback,tv:microformal}), where $H_1$ and $H_2$ are Hamiltonians on $M_1$ and $M_2$ that are the ($\hbar$-)principal  symbols of the operators $\Delta_1$ and $\Delta_2$.

Let $M_1=V_1$ and $M_2=V_2$ be vector spaces (which can also be treated as affine spaces for the purpose of affine transformations).  Consider linear quantum Hamiltonians (including scalar terms)
\begin{equation}\label{eq.deltas}
    \Delta_1= \hat x^aA_a +B^a\hat p_a +K_1\,,\quad
    \Delta_2= \hat y^iC_i +D^i\hat q_i +K_2\,,
\end{equation}
and the   classical Hamiltonians $H_1$ and $H_2$ with the same coefficients. Consider a quantum thick morphism $\hat \F$ with a quadratic action
\begin{equation}\label{eq.quadact}
    S(x,q)= s_0 +x^aS_a +S^iq_i + \frac{1}{2}\,x^ax^bS_{ba}+ x^aS_a^iq_i + \frac{1}{2}\,S^{ij}q_jq_i\,
\end{equation}
and explore the conditions given by the intertwining relation~\eqref{eq.intertwin} and its classical limit~\eqref{eq.hamjacob}.

\begin{theorem} For operators $\Delta_1$ and $\Delta_2$ given by~\eqref{eq.deltas} and  a quantum thick morphism with an action~\eqref{eq.quadact}, the intertwining relation~\eqref{eq.intertwin} and  the Hamilton--Jacobi equation~\eqref{eq.hamjacob} are both equivalent    to the following system of equations
\begin{align}\label{eq.relations1}
    B^aS_a +K_1&= S^iC_i +K_2\,,\\
    A_a+(-1)^{\bt(\e+1)}S_{ab}B^b&= S_a^iC_i\,, \label{eq.relations2}\\
    B^aS_a^j&=(-1)^{\itt(\e+1)}C_iS^{ij}+D^j\,  \label{eq.relations3}
\end{align}
for the coefficients of $\Delta_1$, $\Delta_2$ and $S(x,q)$. \emph{(Here   $\e$ stands for the parity of $\Delta_1$, $\Delta_2$.)}
\end{theorem}
\begin{proof} By a direct calculation.
\end{proof}

Note that there are no constraints on the coefficients $K_1$, $B^a$ and $C_i$ imposed by equations~\eqref{eq.relations1}, \eqref{eq.relations2}, and \eqref{eq.relations3}, so  $K_2$, $A_a$ and $D^i$ are completely determined by $K_1$, $B^a$ and $C_i$ and the coefficients of the action $S$.

We can fix $\hat \F$ defined by an action~\eqref{eq.quadact} and consider the intertwining relation~\eqref{eq.intertwin} as an equation for pairs of linear Hamiltonians $(\Delta_1,\Delta_2)$.

\begin{corollary}
For a given $\hat \F$  with a quadratic action as above, the    intertwining relations
\begin{align}\label{eq.interp}
(-\hat x^aS_{ab}(-1)^{\bt}+ \hat p_b)\circ \hat \F^*&= \hat \F^*\circ (S_b^i\hat q_i+S_b)\,,\\
\hat x^aS_a^i\circ \hat \F^*&= \hat \F^*\circ(\hat y^i-S^{ij}\hat q_j(-1)^{\itt}-S^i)\,, \label{eq.intery}
\end{align}
hold for all $b$ and $i$.  Conversely, pairs of operators in the left-hand   and right-hand sides of~\eqref{eq.interp} and \eqref{eq.intery},  together with   $\Delta_1=\Delta_2=\const$,  give a basis in the space of all pairs of linear Hamiltonians $(\Delta_1,\Delta_2)$ that satisfy  the intertwining relation~\eqref{eq.intertwin} with a given $\hat\F$.
\end{corollary}

\begin{corollary}
The above is equivalent to the formulas
\begin{align}\label{eq.plag}
    p_b&=x^aS_{ab}(-1)^{\bt} +S_b^iq_i+S_b\,,\\
    y^i&=x^aS_a^i+S^{ij}q_j(-1)^{\itt} +S^i \label{eq.play}
\end{align}
holding on the Lagrangian submanifold in $T^*V_1\times T^*V_2$ corresponding to $\hat \F$.
\end{corollary}

Observe that the constant term $s_0$ in the   action $S(x,q)$ cannot be found from here. This corresponds to the fact that the intertwining relation  can   determine $\hat \F$    only up to a   factor.

Consider now three vector spaces, $V_1$, $V_2$ and $V_3$ with the  cotangent bundles $T^*V_1\cong W_1= V_1\oplus V_1^*$, $T^*V_2\cong W_2= V_2\oplus V_2^*$ and $T^*V_3\cong W_3= V_3\oplus V_3^*$.

We want to compare the ``classical'' composition of thick morphisms   with the ``quantum'' composition.

Consider quantum thick morphisms $\hat \F_{32}\co V_2\ttoh V_3$ and $\hat \F_{21}\co V_1\ttoh V_2$ with quadratic actions $S_{32}(y,r)$ and $S_{21}(x,q)$ respectively. (Here $x^a,p_a$ are positions and momenta on $V_1$, $y^i,q_i$ on $V_2$, and $z^{\mu}, r_{\mu}$ on $V_3$.) Suppose
\begin{equation}\label{eq.s21}
    S_{21}(x,q)=s_0 +x^aS_a +S^iq_i + \frac{1}{2}\,x^ax^bS_{ba}+ x^aS_a^iq_i + \frac{1}{2}\,S^{ij}q_jq_i\,
\end{equation}
and
\begin{equation}\label{eq.s32}
    S_{32}(y,r)= {t_0} +y^i {T_i}  + {T^{\mu}}r_{\mu} + \frac{1}{2}\,y^iy^j {T_{ji}}+ y^i {T_i^{\mu}}r_{\mu} + \frac{1}{2}\, {T^{\mu\nu}}r_{\nu}r_{\mu}\,.
\end{equation}

\begin{theorem}
The action for the composition of quantum thick morphisms
\begin{equation}\label{eq.quantcomp}
   \hat \F_{31}:= \hat \F_{32}\circ \hat \F_{21} \co V_1\ttoh V_3
\end{equation}
has the form
\begin{equation}\label{eq.quantcompact}
    S_{31} = S_{31}^{\text{\emph{class}}}  -\frac{\hbar}{i}\,c\bigl(\hat \F_{32}, \hat \F_{21}\bigr) \,,
\end{equation}
where $S_{31}^{\text{\emph{class}}}$ is given by the ``classical'' composition formula
\begin{equation}\label{eq.classcomp}
    S_{31}^{\text{\emph{class}}}(x,r)= S_{32}(y,r) + S_{21}(x,q) -y^iq_i\,,
\end{equation}
and
\begin{equation}\label{eq.quantphase}
    c\bigl(\hat \F_{32}, \hat \F_{21}\bigr)=\frac{1}{2}\ln\Ber\bigl(\delta_i^j-T_{ik}S^{kj}(-1)^{\kt}\bigr)
\end{equation}
is a ``quantum correction''.
\end{theorem}
\begin{proof} By the general formula for the composition of quantum thick morphisms, see~\cite{tv:microformal},
\begin{equation}\label{eq.compos}
    e^{\frac{i}{\hbar}S_{31}(x,r)}=\int \Dbar(y,q)\, e^{\frac{i}{\hbar}\left(S_{32}(y,r)+S_{21}(x,q)-yq\right)}\,.
\end{equation}
To calculate the integral in our particular case, we can use the fact that it is Gaussian in the variables $y,q$ and that for Gaussian integrals the main term in the stationary phase formula gives the whole answer. For our integral it is just  the value of the exponential at the critical point divided by the square root of the Hessian (since  the  numerical factors are conveniently subsumed in the element of integration $\Dbar(y,q)$, see formulas in the Appendix in~\cite{tv:microformal}). Therefore
\begin{equation}\label{eq.compos2}
    e^{\frac{i}{\hbar}S_{31}(x,r)}=\left(\Ber d^2F\right)^{-1/2}\, e^{\frac{i}{\hbar}\left(S_{32}(y,r)+S_{21}(x,q)-yq\right)}\,,
\end{equation}
where by $F$ we have denoted $S_{32}(y,r)+S_{21}(x,q)-yq$ as   function of $y,q$ and it should be evaluated at the critical point.  Setting the derivatives $\lder{F}{y^a}$ and $\lder{F}{q_a}$ to zero, we obtain the following linear system from which $y$ and $q$ are to be determined as functions of $x$ and $r$\,:
\begin{align}\label{eq.syst1}
    q_i-y^jT_{ji}&=T_i+T_i^{\mu}r_{\mu}\,,\\
    -S^{ij}q_j+y^i&=S^i+x^aS_a^i\,. \label{eq.syst2}
\end{align}
The expression $F=S_{32}(y,r)+S_{21}(x,q)-yq$ into which  $y$ and $q$ are substituted as the   solution of the system \eqref{eq.syst1}, \eqref{eq.syst2} is by the definition the generating function of the composition of the \emph{classical} thick morphisms corresponding to $S_{32}$ and $S_{21}$; we denote it $S_{31}^{\text{\rm{class}}}(x,r)$. The matrix of the system \eqref{eq.syst1}, \eqref{eq.syst2} is basically the Hessian matrix $d^2F$. It is easy to see that its Berezinian up to a sign is $\Ber\bigl(\delta_i^j-T_{ik}S^{kj}(-1)^{\kt}\bigr)$. (The  matrix $(\delta_i^j-T_{ik}S^{kj}(-1)^{\kt}\bigr)$ arises itself in connection with the system \eqref{eq.syst1}, \eqref{eq.syst2};   the solution of the system is expressed via   its inverse.) Hence
\begin{equation}\label{eq.compos3}
    e^{\frac{i}{\hbar}S_{31}(x,r)}=\Bigl(\Ber\bigl(\delta_i^j-T_{ik}S^{kj}(-1)^{\kt}\bigr)\Bigr)^{-\frac{1}{2}}\, e^{\frac{i}{\hbar}S_{31}^{\text{\rm{class}}}(x,r)}\,.
\end{equation}
By taking logarithms and multiplying through by $\frac{\hbar}{i}$, we arrive at~\eqref{eq.quantcompact}.
\end{proof}

\begin{remark}
Denoting $S_{11}^{(32)}:=(T_{ij})$ and $S_{22}^{(21)}:=((-1)^{\itt}S^{ij})$,  we  can re-write
\begin{equation}\label{eq.quantphase1}
    c\bigl(\hat \F_{32}, \hat \F_{21}\bigr)= \frac{1}{2}\,\ln\Ber\bigl(1-S_{11}^{(32)} S_{22}^{(21)}\bigr)\,,
\end{equation}
and  it can be expressed also as
\begin{equation}\label{eq.quantphasetrace}
    c\bigl(\hat \F_{32}, \hat \F_{21}\bigr)=\frac{1}{2}\str\ln\bigl(1-S_{11}^{(32)} S_{22}^{(21)}\bigr)\,,
\end{equation}
by using   the Liouville formula $\Ber e^X=e^{\str X}$ (see e.g.~\cite{tv:ber}). Note $\str\ln(1-A)=\sum_{n\geq 1}\frac{1}{n}\str A^n$.
\end{remark}

\begin{remark} The quantum correction  changes only the constant term in the action and  introduces an extra phase factor $e^{-\frac{\hbar}{i}c\bigl(\hat \F_{32}, \hat \F_{21}\bigr)}$ into the operator $\hat\F_{31}$, but it does not change the underlying linear relation. Note that we can allow the functions $S_{32}$ and $S_{21}$ to depend on $\hbar$ as formal power series. Then the relations corresponding to them will also depend on $\hbar$, but will still be considered as ``classical'' objects having   the usual   composition law.
\end{remark}

We conclude that we have arrived at a  \emph{projective  representation} of the category of linear canonical relations by   pullbacks by quantum thick morphisms, which can be viewed as a ``quantization'' of these canonical relations. It  generalizes the spinor representation of the symplectic and orthogonal groups discussed in the previous subsection\footnote{Strictly speaking, our construction directly generalize the pseudo-orthogonal case with the  signature  $(m,m)$,   admitting real Lagrangian subspaces.}

To elaborate the comparison: the intertwining relation~\eqref{eq.intertwin} is a replacement of the adjoint action $\hat g\mapsto T_{\hat g}$; to the homomorphism from the spinor group to the symplectic or orthogonal group here corresponds  the map sending a quantum thick morphism with a quadratic generating function $S(x,q)$ to the underlying canonical relation; and the spinor representation is the inverse map, i.e. reconstructing $S(x,q)$ from the relation. Roughly, the direct map (the analog of $\Spin(V)\to \SO(V)$) assigns to $\hat \F$ with a quadratic $S(x,q)$ as in~\eqref{eq.quadact} the matrix of the coefficients appearing in~\eqref{eq.interp},\eqref{eq.intery},\eqref{eq.plag}, \eqref{eq.play}\,,
\begin{equation}
    \hat \F \mapsto \begin{pmatrix}
     S_{ab}(-1)^{\bt} & S_b^i & S_b\\
    S_a^i& S^{ij}(-1)^{\itt}& S^i
\end{pmatrix}\,.
\end{equation}
Clearly, the inverse map to that (which is the spinor representation) cannot be single-valued because the constant term $s_0$ in $S(x,q)$  is undefined. It can be made single-valued however, if we consider (as we do for classical thick morphisms) not just canonical relations, but ``framed''   relations, i.e. basically fixing the  constants. The so defined representation (from ``framed'' canonical relations to quantum pullbacks, i.e. integral operators) is nevertheless projective because of the ``quantum correction''~\eqref{eq.quantcompact}.

The possibility of extending the spinor representation of the orthogonal and symplectic groups  as constructed by Berezin in~\cite{berezin:second}  to categories of linear relations was discovered by Neretin~\cite{neretin:spinor1989, neretin:categories}, see also  book~\cite{neretin:categories-book}. We   see that quantum thick morphisms specialized to the linear case lead naturally to an analog of the Berezin--Neretin construction. It is interesting to compare these constructions in greater detail.  One can notice that  action for a thick morphism  $S=S(x,q)$  as  a function of position on the source manifold and momentum on the target manifold, corresponds to    ``Potapov--Ginzburg transform''  of~\cite{neretin:spinor1989, neretin:categories, neretin:categories-book} \footnote{The terminology originating from functional analysis, see~\cite{azizov:iokhvidov1986}}. Spinor representation in the Berezin--Neretin approach is based on holomorphic realization of Fock spaces and on normal (Wick) ordering. Our formulas for quantum thick morphisms  generalize the  $xp$-quantization.
Following that, it should be interesting to look at description  of quantum thick morphisms using other types of action.


\def\cprime{$'$}

\end{document}